\documentclass[runningheads]{llncs}
\usepackage[dvipsnames]{xcolor}
\usepackage{graphicx}
\usepackage{algorithm,algpseudocode}
\usepackage{algorithm}
\usepackage{algpseudocode}
\usepackage[linesnumbered,ruled,vlined,algo2e]{algorithm2e}
\usepackage{booktabs}
\usepackage{tabularx}
\usepackage{gensymb}
\usepackage{enumerate}
\usepackage{hyperref}
\usepackage{multirow}
\usepackage{xcolor}
\usepackage[labelformat=simple]{subcaption}

\usepackage{mathrsfs}
\usepackage{comment} 
\usepackage{multicol}
\usepackage{graphicx,tipa}
\usepackage{amsfonts}
\usepackage{amsmath,pdfpages}
\usepackage{tcolorbox}

\begin{document}
\title{Pebble guided Treasure Hunt in Plane\thanks{The authors of the paper are thankful to Dr. Yoann Dieudonné for his valuable inputs and suggestions.}}
%
\author{Adri Bhattacharya\inst{1}\orcidID{0000-0003-1517-8779} \and
Barun Gorain\inst{2} \and
Partha Sarathi Mandal\inst{1}\orcidID{0000-0002-8632-5767}}
\authorrunning{Bhattacharya et al.}
%
\institute{Indian Institute of Technology Guwahati, India\\ \email{\{a.bhattacharya, psm\}@iitg.ac.in} \and
Indian Institute of Technology Bhilai, India\\
\email{barun@iitbhilai.ac.in}}
\maketitle

\begin{abstract}

We study the problem of treasure hunt in a Euclidean plane by a mobile agent with the guidance of pebbles.
The initial position of the agent and position of the treasure are modeled as special points in the Euclidean plane. The treasure is situated at a distance at most $D>0$ from the initial position of the agent. The agent has a perfect compass, but an adversary controls the speed of the agent. Hence, the agent can not measure how much distance it traveled for a given time.
The agent can find the treasure only when it reaches the exact position of the treasure. The cost of the treasure hunt is defined as the total distance traveled by the agent before it finds the treasure. The agent has no prior knowledge of the position of the treasure or the value of $D$. An Oracle, which knows the treasure's position and the agent's initial location, places some pebbles to guide the agent towards the treasure. Once decided to move along some specified angular direction, the agent can decide to change its direction only when it encounters a pebble or a special point. 

We ask the following central question in this paper:

``For given $k \ge 0$, What is cheapest treasure hunt algorithm if at most $k$ pebbles are placed by the Oracle?"

We show that for $k=1$, there does not exist any treasure hunt algorithm that finds the treasure with finite cost. We show the existence of an algorithm with cost $O(D)$ for $k=2$. 
For $k>8$  
 we have designed an algorithm that uses $k$ many pebbles to find the treasure  with cost $O(k^{2}) + D(\sin\theta' + \cos\theta')$, where $\theta'=\frac{\pi}{2^{k-8}}$. The second result shows the existence of an algorithm with cost arbitrarily close to $D$ for sufficiently large values of $D$.
\newline
{ \noindent{\bf Keywords} Treasure Hunt, Mobile agent, Pebbles, Euclidean plane, Deterministic algorithms.}

\end{abstract}

\section{Introduction}

Treasure Hunt problem is the task of finding an inert target by a mobile agent in an unknown environment. The unknown environment can be modeled as a network or a plane. Initially placed at a point in the unknown environment, a mobile agent has to find an inert target, called the treasure. The target or the treasure can be a  miner lost in a cave. The cave can be uninhabitable for humans to search for the lost miner, or it can be inundated with toxic waters henceforth, the person should be found as fast as possible. In computer science applications, a software agent must visit the computers connected by a local area network to find the computer affected by malware. 

In this paper, we study the problem of treasure hunt in the Euclidean plane under a very weak scenario which assumes very little knowledge and control power of the mobile agent. Specifically, the agent does not have any prior knowledge about the position of the treasure or its distance from the treasure. Moreover, the agent has no control over the speed of its movement, and it is assumed that an adversary completely controls the speed of the agent. In practice, for software agents in a network, the movement speed of the agent depends on various factors, such as congestion in the network. In the case of hardware mobile robots, their speeds depend on many mechanical characteristics as well as environmental factors. The agent is equipped with a perfect compass, which helps the agent to rotate and move in a prescribed direction. The agent is initially placed at a point $P$ in the plane. The treasure $T$ is located at most $D >0$ distance (unknown to the agent) from $P$. The agent finds the treasure only when it reaches the exact position of the treasure. The agent's initial position is considered a special point, and the agent can detect this point whenever it visits $P$.  

In the absence of control over its movement speed, once the agent decides to move along a particular angle, it is very important for the agent to learn when to stop its movement. Otherwise, the adversary can increase the speed arbitrarily high, and the agent ends up traversing an arbitrarily large distance. In order to enable the agent to have some control over its movement, an Oracle, knowing the position of the treasure, and the initial position of the agent, places some stationary pebbles on the plane. We assume a restriction on the pebble placement by the Oracle: any two pebbles must separated by a constant distance, i.e., no two pebbles are placed arbitrarily close \footnote{This is required if the sensing capability of the agent is weak, two pebbles placed very close to each other may not be distinguished by the agent. }.. For simplicity, we assume that  any two pebbles must be placed at least 1 distance apart. The agent can detect the existence of a pebble only when it reaches the position of the pebble where its placed by the Oracle.

These pebbles placement helps the agent control its movement and rule out the possibility of traversing arbitrarily large distances. Starting from some position of the plane, the agent, moving along a specific direction, stops or changes its direction once it encounters a pebble along the path of its movement. Thus, the movement algorithm of the agent gives instruction to move along a specific angle $\alpha$ until it encounters a special point (i.e., the initial position $P$ or the position of the treasure $T$) or it hits a pebble.  

Formally, for a given positive integer $k\ge 0$, the Oracle is a function $f_k:(E \times E) \rightarrow  E^k$, where $E$ is the set of of all the points in the Euclidean Plane. The function takes two points as the input, the first one is the initial position of the agent, and the second one is the position of the treasure, and gives $k$ points in the plane as output which represents the placement of a pebble at each of these $k$ points. 

The central question studied in this paper is: ``For given $ k \ge 0$, what is the minimum cost of treasure hunt if at most $k$ pebbles are placed in the plane?"

\subsection{Contribution}
Our contributions in this paper are summarized below.
\begin{itemize}

    \item For $k=1$ pebbles, we have shown that it is not possible to design a treasure hunt algorithm that finds the treasure with finite cost. 
    \item For $k=2$ pebbles, we propose an algorithm that finds the treasure with cost at most $4.5 D$, where $D$ is the distance between the initial position of the agent and the treasure.
    \item For $k>8$, we design an algorithm that finds the treasure using $k$ pebbles with cost $O(k^{2}) + D\left(\sin\theta' + \cos\theta'\right)$, where $\theta' = \frac{\pi}{2^{{k-8}}}$. For sufficiently large values of $D$ and $k \in o (\sqrt{D})$, the cost of this algorithm is arbitrarily close to $D$, the cost of the optimal solution in case the position of the treasure is known to the agent.

\end{itemize}

\subsection{Related Work}

The task of searching for an inert target by a mobile agent has been rigorously studied in the literature under various scenarios. The underlying environment or the topology may be either a discrete or continuous domain, i.e., a graph or a plane. The search strategy can be either deterministic or randomized. The book by Alpern et al. \cite{alpern2006theory} discusses the randomized algorithms based on searching for an inert target as well as the rendezvous problem, where the target and the agent are both dynamic, and they cooperate to meet. The papers by Miller et al.\cite{miller2015tradeoffs}, and Ta-Shma et al.\cite{ta2007deterministic} relate the correlation between rendezvous and treasure hunt problem in graph.

The problem of searching on a line for an inert target was first initiated by Beck et al. \cite{beck1970yet}. They gave an optimal competitive ratio of $9$. Further, Demaine et al. \cite{demaine2006online} modified the problem, in which there is a cost involved for each change in the direction of the searcher. In \cite{gal2010search}, the author surveys the searching problem in terms of search games where the target is either static or mobile. The search domain is either a graph, a bounded domain, or an unbounded set. Fricke et al. \cite{fricke2016distributed}, generalized the search problem in a plane with multiple searchers.

Now, the paradigm of \textit{algorithm with advice} has been introduced mainly in networks. These {\it advice} enhances the efficiency of the problems as discussed in \cite{abiteboul2006compact,bhattacharya2022treasure,bouchard2020deterministic,bouchard2022impact,dereniowski2012drawing,emek2011online,fraigniaud2009distributed,fraigniaud2008tree,fraigniaud2010communication,fraigniaud2007local,fusco2008trade,gorain2022pebble,pelc2019cost}. In this paradigm, the authors \cite{fraigniaud2010communication,fraigniaud2007local} mainly studied the minimum amount of advice required in order to solve the problem efficiently. In \cite{dobrev2012online,fraigniaud2009distributed}, the online version of the problems with advice was studied. The authors \cite{bouchard2020deterministic}, considered the treasure hunt problem, in which they gave an optimal cost algorithm where the agent gets a hint after each movement. Pelc et al. \cite{pelc2021advice}, gave an insight into the amount of information required to solve the treasure hunt in geometric terrain at $O(L)$- time, where $L$ is the shortest path of the treasure from the initial point. Bouchard et al. \cite{bouchard2022impact}, studied how different kinds of initial knowledge impact the cost of treasure hunt in a tree network.

The two papers closest to the present work are \cite{gorain2022pebble,pelc2019cost}. Pelc et al. \cite{pelc2019cost}, provided a trade-off between cost and information of solving the treasure hunt problem in the plane. They showed optimal and almost optimal results for different ranges of vision radius. Gorain et al. \cite{gorain2022pebble}, gave an optimal treasure hunt algorithm in graphs with pebbles, termed as advice. In \cite{bhattacharya2022treasure}, the authors studied a trade-off between the number of pebbles vs. the cost of the treasure hunt algorithm in an undirected port-labeled graph.

\noindent\textbf{Organization:} The paper is organized in the following manner. Section \ref{feasibility} gives a brief idea about the feasibility of the treasure hunt problem when a certain number of pebbles are placed. Section \ref{FasterTH} is subdivided into three subsections, in subsection \ref{highlevelidea}, the high-level idea of the algorithm is described, in subsection \ref{pebbleplacement}, the pebble placement strategy is described, and in subsection \ref{treasurehunt} the treasure hunt algorithm is discussed. In section \ref{correct}, correctness and complexity are discussed. Further, in section \ref{section-4}, possible future work and conclusion are explained.

\section{Feasibility of Treasure hunt}\label{feasibility}

In this section, we discuss the feasibility of the treasure hunt problem, when the oracle places one and two pebbles, respectively.

\begin{theorem}
It is not possible to design a treasure hunt algorithm using at most one pebble that finds the treasure at a finite cost.
\end{theorem}
\begin{proof}
The agent initially placed at $P$ and the pebble is placed somewhere in the plane by the oracle. Since the agent has no prior information about the location of the treasure, the treasure can be positioned anywhere in the plane by the adversary. The only possible initial instruction for the agent is to move along a certain angle from $P$. The agent along its movement, must encounter a pebble otherwise, it will continue to move in this direction for an infinite distance, as it has no sense of distance. After encountering the pebble, there are three possibilities: either it may return back to $P$ and move at a certain angle from $P$ or it may return back to $P$ and move along the same path traversed by the agent previously to reach the pebble or it may move at a certain angle from the pebble itself. The adversary may place the treasure at a location different from the possible path to be traversed by the agent. Hence, it is not possible to find the treasure at a finite cost.

\qed
\end{proof}

In this part, we discuss the strategy of pebble placement and respective traversal of the agent toward the treasure when two pebbles are placed by the oracle.

\noindent\textbf{Pebble Placement:}  Based on the location of the treasure, two pebbles are placed as follows. Let the coordinates of the treasure $T$ be $(x_T,y_T)$. If either of $x_T$ or $y_T$ is positive, place one pebble at $(z+1,z+1)$, where $z=\max\{|x_T|,|y_T|\}$. Place another pebble at $(x_T,z+1)$. Otherwise, if both $x_T$ and $y_T$ are negative, place one pebble at $(1,1)$ and another pebble at $(x_T,1)$.

\noindent\textbf{Treasure Hunt by the agent:} The agent initially at $P$, moves at an angle $\frac{\pi}{4}$ with the positive $x$ axis until it encounters treasure or a pebble (i.e., $p_1$). If a pebble is encountered, then from this position the agent moves along an angle $\pi-\frac{\pi}{4}$ until it encounters the treasure or reaches a pebble (i.e., $p_2$). If a pebble is encountered (i.e., from $p_2$), the agent further moves along an angle $\frac{\pi}{2}$ until it reaches the treasure $T$.
 
\begin{theorem}
The agent finds the treasure with cost $O(D)$ using the above algorithm.
\end{theorem}

\begin{proof}

According to the proposed algorithm, the cost of finding the treasure is the path $Pp_1+p_1p_2+p_2T$ (refer Fig. \ref{two-12}), where $p_1$ and $p_2$ are the positions of the first and second pebbles, respectively. Let $f_i : \theta \longrightarrow \mathbb{R}$, where $i= 1, \cdots, 5$, be the set of cost functions for each of the following cases, we analyze them as follows:

\begin{itemize}
    \item[1:] If the treasure is on the first quadrant, then let $A$ and $B$ be the foot of the perpendicular drawn from $T$ and $p_1$, respectively. Let $\angle TPA =\theta$ (refer Fig. \ref{two}). So, $PA=D\cos\theta$ and $AT=D\sin\theta$. Now we have the following cases:

\begin{figure}
\centering
\begin{subfigure}{.5\textwidth}
  \centering
  \includegraphics[width=.9\linewidth]{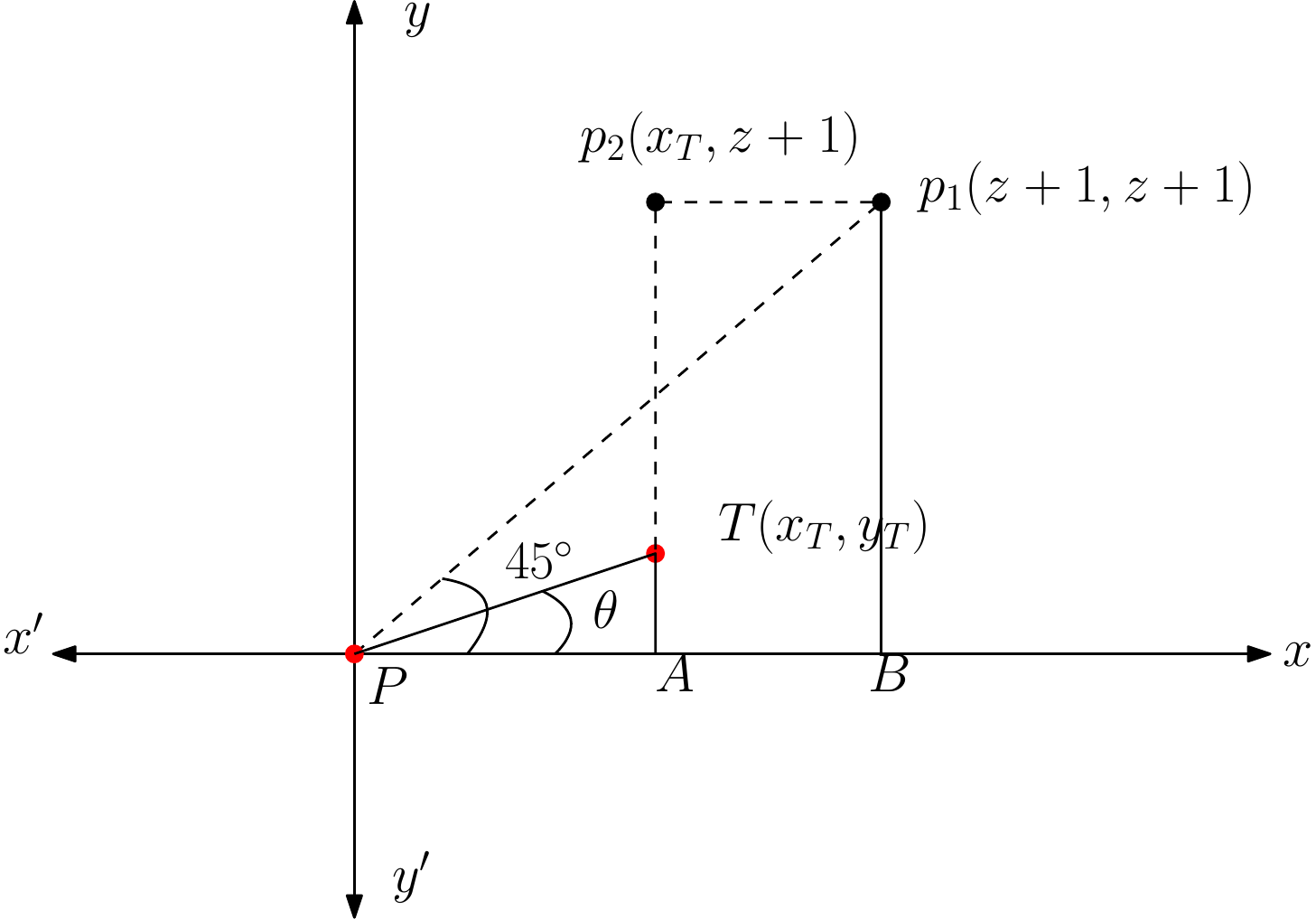}
  \caption{Treasure in $1$st Quadrant}
  \label{two}
\end{subfigure}%
\begin{subfigure}{.5\textwidth}
  \centering
  \includegraphics[width=.9\linewidth]{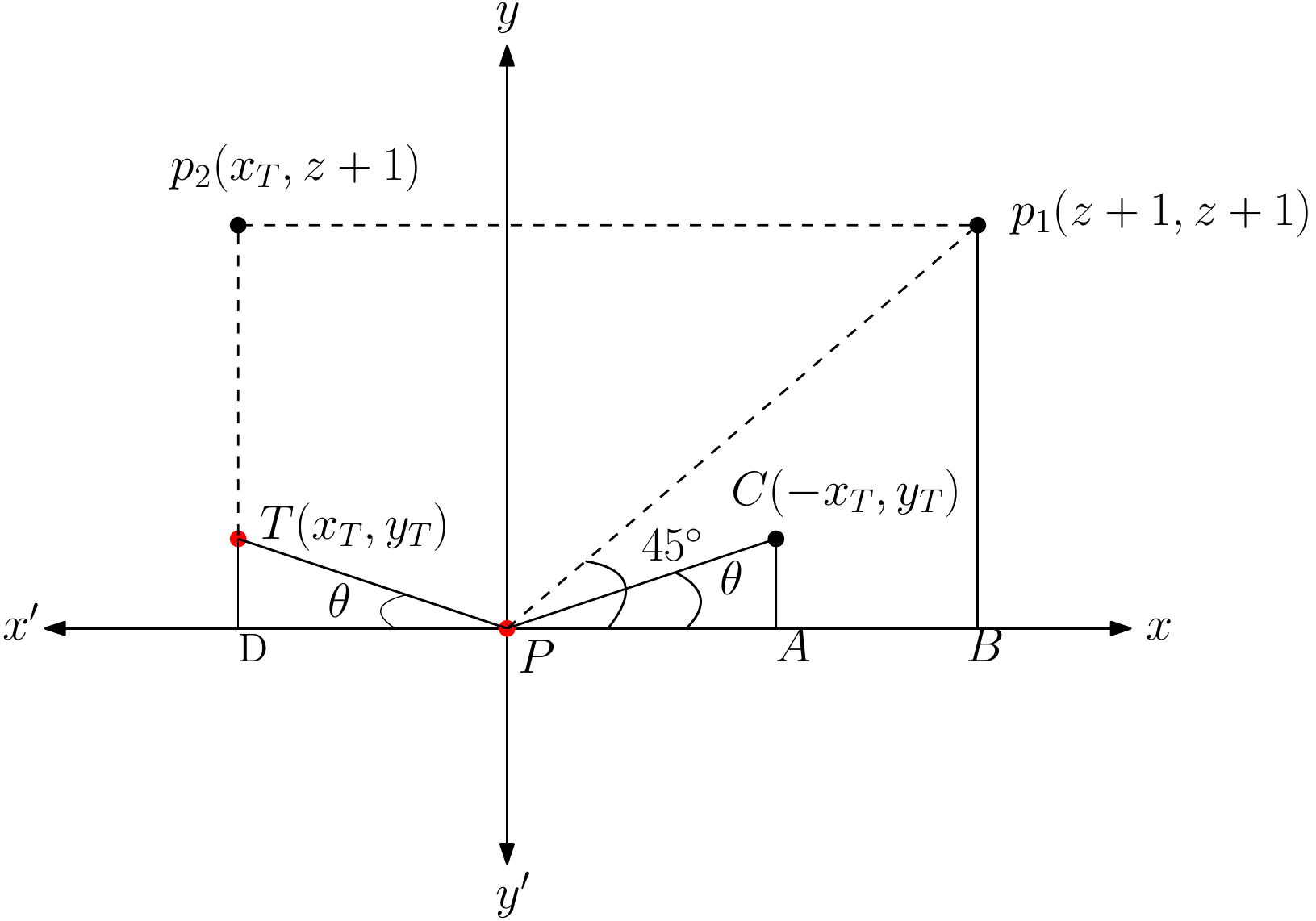}
  \caption{Treasure in $2$nd Quadrant}
  \label{two-2nd}
\end{subfigure}
\caption{Movement of the agent when the treasure is located in the upper half of the plane}
\label{two-12}
\end{figure}

    \begin{itemize}
        \item [1(a):] When $x_T\geq y_T$, then the pebbles $p_1$ and $p_2$ are placed at $(x_T+1,x_T+1)$ and $(x_T,x_T+1)$, respectively. So, $PB=D\cos\theta+1$ and $PB=Bp_1$ (since $\Delta p_1PB$ is an isosceles triangle), this implies $Pp_1= \sqrt{2}(D\cos\theta+1)$. Moreover in this case, $p_1p_2=1$ and $p_2T=p_2A-TA=D\cos\theta+1-D\sin\theta$. So, the total cost is: $\sqrt{2}(D\cos\theta+1)+1+(D\cos\theta+1-D\sin\theta)$ which is linear in terms of $D$. 
        \item [1(b):] When $y_T> x_T$, then the pebbles $p_1$ and $p_2$ are placed at $(y_T+1,y_T+1)$ and $(x_T,y_T+1)$, respectively. So, $Bp_1=D\sin\theta+1$ and $PB=Bp_1$ (since $\Delta p_1PB$ is an isosceles triangle), this implies $Pp_1= \sqrt{2}(D\sin\theta+1)$. Moreover in this case $p_1p_2=(D\sin\theta+1)-D\cos\theta$ and $p_2T=p_2A-TA=D\sin\theta+1-D\sin\theta=1$. So, the total cost is: $\sqrt{2}(D\sin\theta+1)+(D\sin\theta+1)-D\cos\theta+1$ which is again linear in terms of $D$. 
                
    \end{itemize}
    So, $f_1(\theta) = \max\{\min_{\forall \theta \in [0,2\pi]}\{\sqrt{2}(D\cos\theta+1)+1+(D\cos\theta+1-D\sin\theta)$, $\sqrt{2}(D\sin\theta+1)+(D\sin\theta+1)-D\cos\theta+1\}\}$
    \item[2:] If the treasure is on the second quadrant, let C be the mirror image of $T$ on first quadrant (refer Fig. \ref{two-2nd}) then consider $A$ and $B$ be the foot of the perpendicular drawn from $C$ and $p_1$, respectively. Let $\angle TPD =\theta$, and hence $\angle CPA=\theta$. So, we have $PA=D\cos\theta$ and $AC=D\sin\theta$. We have the following cases:
    
    \begin{itemize}
        \item [2(a):] When $|x_T|\geq y_T$, then the pebbles $p_1$ and $p_2$ are placed at $(-x_T+1,-x_T+1)$ and $(x_T,-x_T+1)$, respectively. So, $PB=D\cos\theta+1$ and $PB=Bp_1$ (since $\Delta p_1PB$ is an isosceles triangle), this implies $Pp_1= \sqrt{2}(D\cos\theta+1)$. Moreover in this case, $p_1p_2=D\cos\theta+1+D\cos\theta$ and $p_2T=p_2A-TA=D\cos\theta+1-D\sin\theta$. So, the total cost is: $\sqrt{2}(D\cos\theta+1)+(D\cos\theta+1-D\sin\theta)+(2D\cos\theta+1)$ which is linear in terms of $D$.
        \item [2(b):] When $y_T> |x_T|$, then the pebbles $p_1$ and $p_2$ are placed at $(y_T+1,y_T+1)$ and $(x_T,y_T+1)$, respectively. So, $Bp_1=D\sin\theta+1$ and $PB=Bp_1$ (since $\Delta p_1PB$ is an isosceles triangle), this implies $Pp_1= \sqrt{2}(D\sin\theta+1)$. We have $p_1p_2=(D\sin\theta+1)+D\cos\theta$ and $p_2T=p_2A-TA=D\sin\theta+1-D\sin\theta=1$. So, the total cost is $\sqrt{2}(D\sin\theta+1)+(D\sin\theta+1)+D\cos\theta+1$, which is again linear in terms of $D$.
        
    \end{itemize}
    So, $f_2(\theta) = \max\{\min_{\forall \theta \in [0,2\pi]}\{\sqrt{2}(D\cos\theta+1)+(D\cos\theta+1-D\sin\theta)+(2D\cos\theta+1), \sqrt{2}(D\sin\theta+1)+(D\sin\theta+1)+D\cos\theta+1\}\}$.
    \item[3:] If the treasure is on the third quadrant, let $A$ and $B$ be the foot of the perpendicular drawn from $p_2$ and $p_1$. Let $\angle TPA =\theta$ (refer Fig. \ref{two-3rd}) and the pebbles $p_1$ and $p_2$ are placed at $(1,1)$ and $(x_T,1)$, respectively. So, $Pp_1=\sqrt{2}$, $p_1p_2=1+D\cos\theta$ and $p_2T=p_2A+AT=1+D\sin\theta$. So, the total cost is: $\sqrt{2}+1+D\cos\theta+D\sin\theta$, which is again linear in terms of $D$. Hence, $f_3(\theta) = \min_{\forall \theta \in [0,2\pi]} \{\sqrt{2}+1+D\cos\theta+D\sin\theta\}$.

\begin{figure}
\centering
\begin{subfigure}{.5\textwidth}
  \centering
  \includegraphics[width=.9\linewidth]{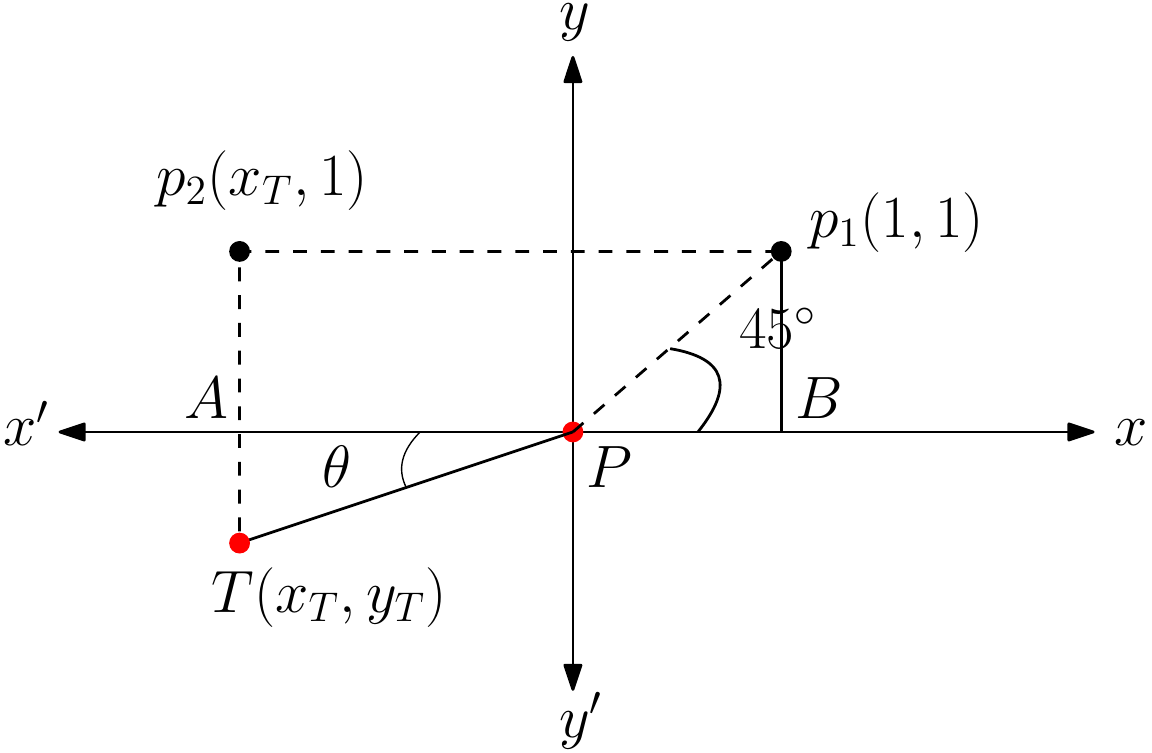}
  \caption{Treasure in $3$rd Quadrant}
  \label{two-3rd}
\end{subfigure}%
\begin{subfigure}{.5\textwidth}
  \centering
  \includegraphics[width=.9\linewidth]{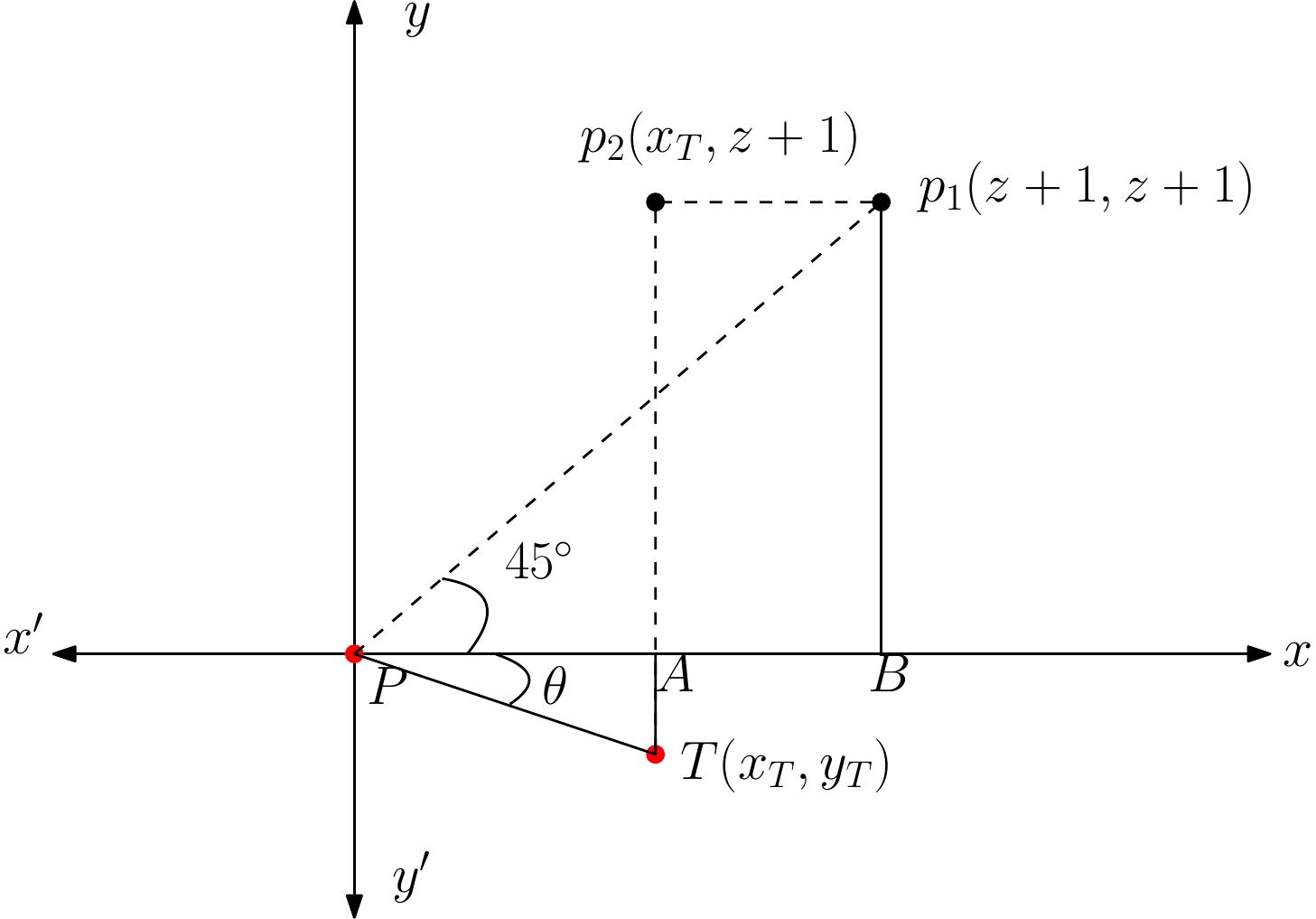}
  \caption{Treasure in $4$th Quadrant}
  \label{two-4th}
\end{subfigure}
\caption{Movement of the agent when the treasure is located at the lower half of the plane}
\label{two-34}
\end{figure}
      
\item[4:] If the treasure is on the fourth quadrant, then let $A$ and $B$ be the foot of the perpendicular drawn from $T$ and $p_1$, respectively. Let $\angle TPA =\theta$ (refer Fig. \ref{two-4th}). So, $PA=D\cos\theta$ and $AT=D\sin\theta$. Now we have the following cases:
\begin{itemize}
        \item [4(a):] When $x_T\geq |y_T|$, then the pebbles $p_1$ and $p_2$ are placed at $(x_T+1,x_T+1)$ and $(x_T,x_T+1)$, respectively. So, $PB=D\cos\theta+1$ and $PB=Bp_1$ (since $\Delta p_1PB$ is an isosceles triangle), this implies $Pp_1= \sqrt{2}(D\cos\theta+1)$. So, we have $p_1p_2=1$ and $p_2T=p_2A+TA=D\cos\theta+1+D\sin\theta$. So, the total cost is: $\sqrt{2}(D\cos\theta+1)+1+(D\cos\theta+1+D\sin\theta)$, which is linear in terms of $D$.
        \item [4(b):] When $|y_T|> x_T$, then the pebbles $p_1$ and $p_2$ are placed at $(-y_T+1,-y_T+1)$ and $(x_T,-y_T+1)$, respectively. So, $Bp_1=D\sin\theta+1$ and $PB=Bp_1$ (since $\Delta p_1PB$ is an isosceles triangle), this implies $Pp_1= \sqrt{2}(D\sin\theta+1)$. Hence, we have $p_1p_2=(D\sin\theta+1)-D\cos\theta$ and $p_2T=p_2A+TA=D\sin\theta+1+D\sin\theta=2D\sin\theta+1$. So, the total cost is: $\sqrt{2}(D\sin\theta+1)+(D\sin\theta+1)-D\cos\theta+2D\sin\theta+1$, which is again linear in terms of $D$.
    \end{itemize}
    So, we have $f_4(\theta) = \max\{\min_{\forall \theta \in [0,2\pi]} \{ \sqrt{2}(D\cos\theta+1)+1+(D\cos\theta+1+D\sin\theta), \sqrt{2}(D\sin\theta+1)+(D\sin\theta+1)-D\cos\theta+2D\sin\theta+1\}\}$
\end{itemize}

Further, the cumulative cost is $f_5(\theta) = \max_{\forall i \in \{1,\cdots,5\}} \{f_i(\theta)\}$, which is approximately $4.5D+(\sqrt{2}+2)$.
Hence, from all the above cases, we conclude that the cost complexity is linear in $D$.

\qed
\end{proof}

\section{Improved solution for treasure hunt}\label{FasterTH}

In this section, we propose faster algorithm which requires at least $9$ pebbles to perform the treasure hunt.

\subsection{High level idea}\label{highlevelidea}

Before we give the details of the pebble placement algorithm, we describe the high-level idea of the same. Intuitively, depending on the number of pebbles available, the Oracle divides the plane into multiple sectors as described in Section \ref{pebbleplacement}. Then it identifies the sector number $m$ in which the treasure is located and `encode' this number by placing the pebbles. The agent, looking at the pebble placements, `decode' this encoding, and move along the particular sector to find the treasure. There are certain challenges that need to be overcome to implement this idea.

\noindent{\bf Sector Detection:} The first difficulty is how different placements of pebbles enable the agent to differentiate between the bit 0 and the bit 1. Since the agent has no sense of time and distance, two different pebble placements may look identical to the agent. On the other hand, since the agent has no prior information about the encoded integer, its movement should also be planned in a way that using the same movement strategy will detect the bit zero for some instances and the bit 1 for other instances. The capability of detecting the initial point $P$ as special point is used to overcome this difficulty. 

First, we place a pebble $p_1$ at the point (1,0) and two additional fixed pebbles $p_2$ at (1,1) and $p_3$ at (2,1) are placed. The rest of the pebbles are placed based on the fact whether a particular bit of the encoding is a 0 or 1. Initially, consider the specific scenario of encoding only one bit 0 or 1. The idea is to place a particular pebble $p$ in two possible positions on the x-axis such that the agent, starting from $P$, reach $p$, then moving at a certain fix angle $\alpha$ from $p$ will reach $p_2$ for one position and $p_3$ for the other. The agent can not distinguish between $p_2$ and $p_3$ but moving in a particular angle $\beta$ from $p_2$ will reach $P$ and from $p_3$ will reach $p_1$. These two different scenarios are distinguished as 0 and 1, respectively. In order to achieve and implement the idea, the pebble $p$ is placed at the point (3,0) in case of encoding 1 and (4,0) in case of encoding 0. The advantage of this specific placement is that in case of placing $p$ at (3,0) that moving from $P$ to $p$, and then moving at an angle $\arctan{(\frac{-1}{2})}$, the agent reaches $p_2$ and then moving at an angle $\arctan{(3)}$, it reaches $P$. On the other hand, in the case of placing $p$ at (4,0), using the same movement, the agent arrives at $p_1$. Hence, it detects these two different observations as two different bits 1 and 0, respectively. (See  Fig. \ref{1 bit}).          

\begin{figure}
\centering
\begin{minipage}{.45\textwidth}
  \centering
  \includegraphics[width=1.0\linewidth]{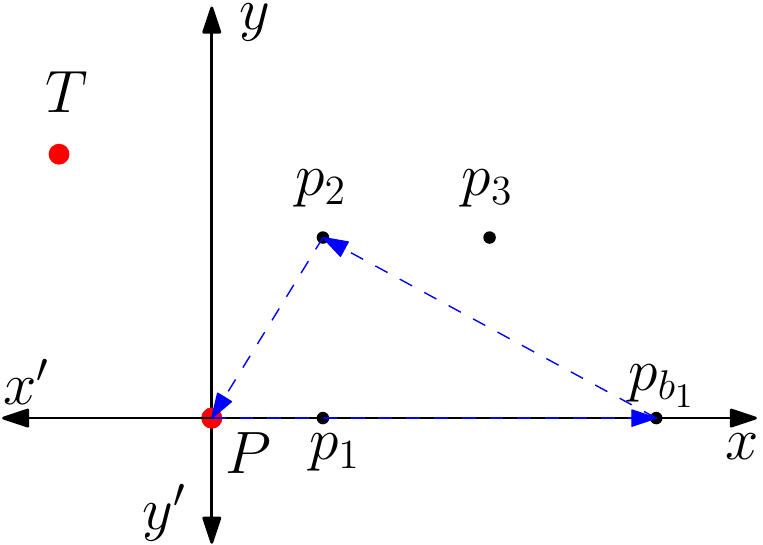}
  \captionof{figure}{Placement of pebble by oracle when first bit is $1$}
  \label{1 bit}
\end{minipage}%
\hfill
\begin{minipage}{.5\textwidth}
  \centering
  \includegraphics[width=.9\linewidth]{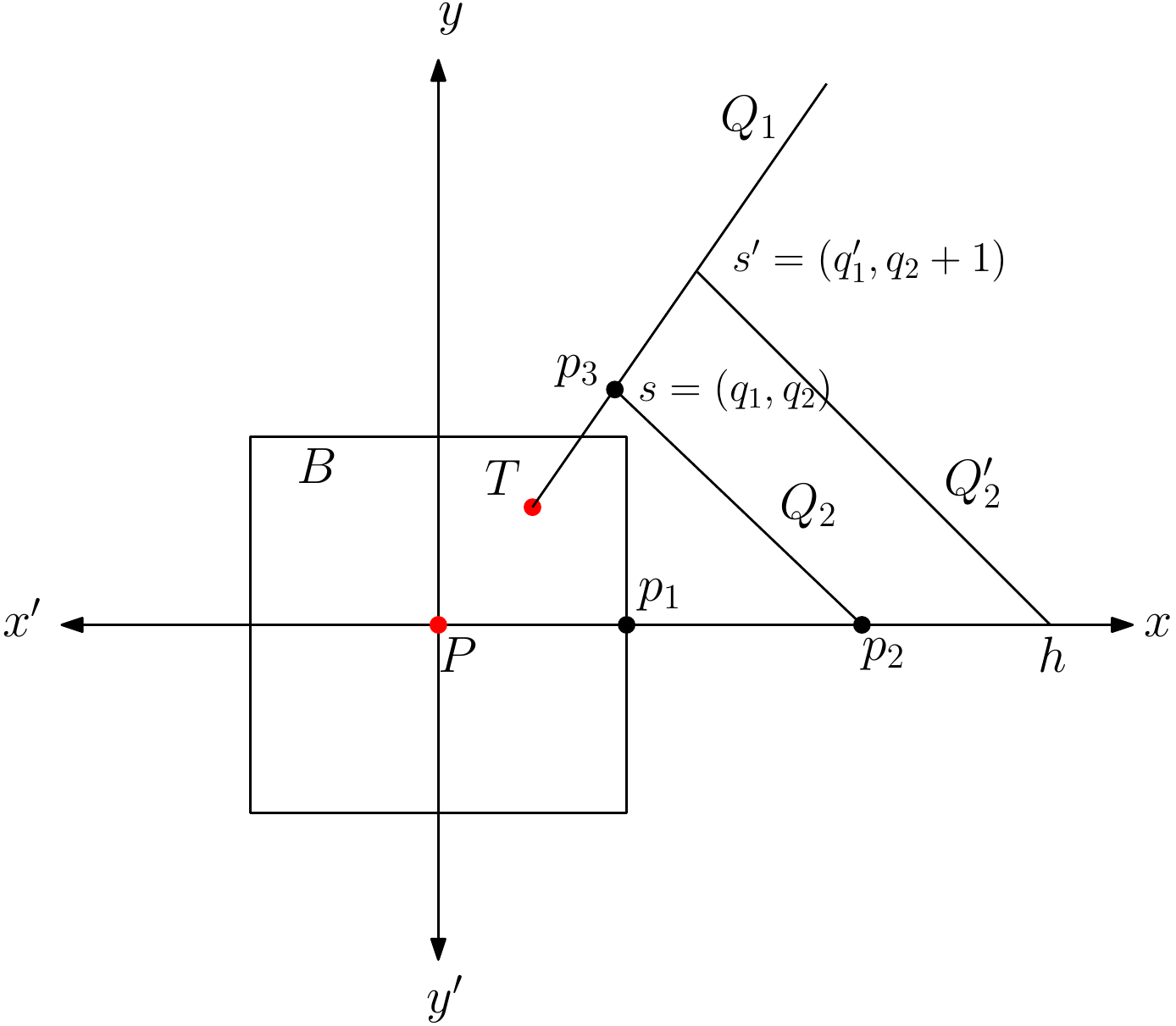}
  \captionof{figure}{Pebble placement when treasure inside $B$}
  \label{square}
\end{minipage}
\end{figure}

We extend the above idea for encoding any binary string $\mu$ as follows. In addition to the pebbles $p_1$, $p_2$, and $p_3$, one additional pebble for each of the bits of $\mu$ are placed. To be specific, for $1 \le i \le |\mu|$, a pebble $p_{b_i}$ is placed at $(2i+1,0)$ if the $i$-th bit is 1, else $p_{b_i}$ is placed at $(2i+2,0)$. Starting from $P$ to $p_{b_i}$, moving at an angle $\arctan\left(\frac{-1}{2i}\right)$ until a pebble is reached, then moving at an angle $\arctan\left(\frac{2i+1}{2i-1}\right)$, the agent reaches either $P$ or $p_1$ depending on the $i$-th bit is 1 or 0 respectively.

A difficulty that remains to be overcome is how the agent detects the end of the encoding. This is important because if the termination is not indicated, then there is a possibility that the agent moves to find more pebbles $p_{b_j}$, $j >|\mu|$, and continues its movement for an infinite distance. We use two additional pebbles named $p_{t_1}$ and $p_{t_2}$ for the specific requirement of termination detection. The position of these two pebbles $p_{t_1}$ and $p_{t_2}$ are as follows. If the 1st bit of the binary string $\mu$ is 1, i.e., $p_{b_1}$ is placed at $(3,0)$ then the pebbles $p_{t_1}$ and $p_{t_2}$ are placed at $(4,1)$ and $(2|\mu|+6,0)$, respectively. Otherwise, if the 1st bit is 0 then these two pebbles are placed at $(5,1)$ and $(2|\mu|+7,0)$, respectively. After visiting the pebble $p_{|\mu|}$ for the last bit of $\mu$, the agent returns to $P$, and moves as usual to find a pebble expecting to learn more bits of the binary string. From $P$, once it reaches $p_{t_2}$, it moves at an angle $\arctan\left(\frac{-1}{2(|\mu|+1)}\right)$ until a pebble  is reached. Note that the two pebbles $p_{t_1}$ and $p_{t_2}$ are placed in such a way that the angle $\angle Pp_{t_2}p_{t_1}=\arctan\left(\frac{-1}{2(|\mu|+1)}\right)$. Hence using the movement from $p_{t_2}$ at angle $\arctan\left(\frac{-1}{2(|\mu|+1)}\right)$ the agent reaches $p_{t_1}$ and from $p_{t_1}$ moving at angle $\arctan\left(\frac{2(|\mu|+1)+1}{2(|\mu|+1)-1}\right)$, it reaches to $p_{b_1}$. Since the following specific movement mentioned above, the agent reaches to a pebble, it initially assumed that it learned the bit zero. But moving west from $p_{b_1}$, it reaches another pebble (i.e., the pebble $p_1$), instead of origin. This special occurrence indicates the termination of the encoding to the agent. Hence in this way, the agent learns the binary string $\mu$, and the integer $\Delta$ whose binary representation is $\mu$. 

\noindent{\bf Finding the treasure inside the sector:} One more pebble  $p_T$ is  placed on the foot of the perpendicular drawn from $T$ on $L_{j+1}$ (refer Fig.\ref{ex-travel}). After learning the encoding of $\mu_j$, the agent decodes the integer $j$, and correctly identifies the two lines $L_j$ and $L_{j+1}$ inside the sector to help the agent in locating the exact location of the treasure.

A difficulty arises here while placing the pebble $p_T$ inside the sector as some pebbles that are already placed while the encoding of the sector number may be very close (at distance $<1$) from the possible position of $p_T$. To resolve this, we do the placement of the pebbles for encoding on positive x-axis if the position of the treasure is at the left half plane, and the placement of the pebbles are done on the negative x-axis, if the position of the treasure is at the right half plane. To instruct the agent which way it must move to find the pebbles for learning the encoding, one additional pebble $p_0$ is placed at $P$.  

Some specific cases need to be separately handled:  If the treasure is in a position $(x,y)$, such that $-1\le x \le 1$ and $y \ge -1$, as this again may create a problem in placing $p_T$ inside the prescribed position inside the sector. The details of these cases will be discussed while explaining the pebble placement strategy in the section \ref{pebbleplacement}.

\subsection{Pebble placement}\label{pebbleplacement}

The agent is initially placed at  $P$, and the treasure is placed at $T$. 
The oracle, knowing the initial position $P$ of the agent and the position $T=(x_T,y_T)$ of the treasure, places 
 $k$ pebbles in the Euclidean plane. Let $B$ be the square region bounded by the lines $x=1$, $x=-1$, $y=1$, and $y=-1$.  
 
 Based on the position of the treasure, the pebble placement is described using two different cases.
 
 \noindent{\bf Case 1:} 
 If $x_T>0$ and $T \not\in B$, then the placements of the pebbles are done as follows.

 \begin{enumerate}
 \item  Place a pebble $p_0$ at $P$.

 \item Draw $2^{k-8}$ half-lines $L_0, \cdots, L_{2^{k-8}-1}$, starting at the initial position $P$ of the agent, such that $L_0$ goes North and the angle between consecutive half-lines is $\pi/2^{k-8}$ for $i = 0, \cdots ,2^{k-8}-1$. The sector $S_i$
is defined as the set of points in the plane between $L_i$ and $L_{(i+1) \mod 2^{k-8}}$
, including the points on  $L_i$ and excluding the points on $L_{(i+1) \mod 2^{k-8}}$. If $T \in S_j$, for some $j \in \{0,1,\cdots,2^{k-8}-1\}$ then place pebbles as follows.
\begin{itemize}
    \item Place the pebbles $p_1$ at (-1,0), $p_2$ at (-1,-1) and $p_3$ at (-2,-1).
    \item Let $\mu_j$ be the binary representation of the integer $j$ with leading $\lfloor \log k\rfloor - \lfloor \log j\rfloor$ many zeros. If $0\le x_T \le 1$ and $y_T>1$, then $\mu_j=0 \cdot \mu_j$, else $\mu_j=1 \cdot \mu_j$.
    For $1\le \ell \le |\mu_j|$, if the $\ell$-th bit of $\mu_j$ is 1, then place a pebble at $(-2\ell-1,0)$, else place a pebble at $(-2\ell-2,0)$.
    \item If the 1st  bit of $\mu_j$ is 1, then place a pebble $p_{t_1}$ at (-4,-1), else place $p_{t_1}$ at (-5,-1).
    \item If the 1st  bit of $\mu_j$ is 1, then place a pebble $p_{t_2}$ at $(-2|\mu_j|-6,0)$, else place $p_{t_2}$ at $(-2|\mu_j|-7,0)$.
\end{itemize}
\item If $x_T<0$ and $T \not\in B$, 
  then the placements of the pebbles are done as follows.
 For each pebble placed at $(m,n)$, where $m\ne 0$ or $n \ne 0$ in the above case, place the corresponding pebble at $(-m,-n)$ in this case. Also, place no pebble at $P$.

\item If the first bit of $\mu_j$ is 0, then let $F$ be the foot of the perpendicular drawn from $T$ to $L_j$, else let $F$ be the foot of the perpendicular drawn from $T$ to $L_{j+1}$. Place a pebble $p_T$ at $F$ ( Lemma \ref{noconflict}  ensures that the pebbles are placed at a distance of at least 1 in this scenario).

\noindent{\bf Case 2:} If $x_T>0$ and $T \in B$, then the pebbles are placed as follows. 
\begin{itemize}
    \item Place a pebble $p_1$ at (1,0) (refer Fig. \ref{square}).
    \item Let $m_1=\tan\left(\pi-\arctan(\frac{-1}{2})-\arctan(3)\right)$ and $m_2=\tan(\pi-\arctan$ $\left(\frac{-1}{2}\right))$. Draw a line $Q_1$ through $T$ with slope $m_1$ and draw a line $Q_2$ through the point $(2,0)$ with slope $m_2$. Let $s=(q_1,q_2)$ be the point of intersection between these two lines. Let $s'$ be the point on the line $Q_1$ whose $y$ coordinate is $q_2+1$. Draw the lines $Q_2'$  parallel to $Q_2$ and go through $s'$. Let $h$  be the points of intersection of the lines $Q_2'$  with $x$-axis.
    
    Two additional pebbles $p_2$ and $p_3$ are placed as follows.   
    
  \begin{itemize}
        \item If $q_2< 1$, then place  $p_2$ at $h$ and  $p_3$ at $s'$.
        \item Otherwise, place  $p_2$ at $(2,0)$  and  $p_3$ at $s$.
 
  \end{itemize}
\end{itemize}
If $x_T <0$ and $T \in B$, then placement of the pebbles are done as follows. 
\begin{itemize}
    \item Place the pebbles $p_0$ at P and $p_1$ at (-1,0).
    \item Let $m_1=-\tan\left(\pi-\arctan(\frac{-1}{2})-\arctan(3)\right)$ and $m_2=-\tan(\pi-\arctan$ $\left(\frac{-1}{2}\right))$. Draw a line $Q_1$ through $T$ with slope $m_1$ and draw a line $Q_2$ through the point $(-2,0)$ with slope $m_2$. Let $r=(r_1,r_2)$ be the point of intersection between these two lines. Let $r'$ be the point on the line $Q_1$ whose $y$ coordinate is $r_2+1$. Draw the lines $Q_2'$  parallel to $Q_2$ and go through $r'$. Let $n$  be the points of intersection of the lines $Q_2'$  with $x$-axis.
    
    Two additional pebbles $p_2$ and $p_3$ are placed as follows.   
    
  \begin{itemize}
        \item If $r_2< 1$, then place  $p_2$ at $n$ and  $p_3$ at $r'$.
        \item Otherwise, place  $p_2$ at $(-2,0)$  and  $p_3$ at $r$.
 
  \end{itemize}
\end{itemize}

\end{enumerate}


\begin{algorithm2e}

\caption{\sc{PebblePlacement}}\label{alg:cap0}

Draw $2^{k-8}$ half lines $L_0,\cdots, L_{2^{k-8}-1}$ starting from $P$, where angle between two consecutive half-lines is $\frac{\pi}{2^{k-8}}$. Let Sector $S_i$ be the sector bounded by the half lines $L_i$ and $L_{i+1}$ and let $T\in S_{\Delta}$, $\Delta \in \{0,1,\cdots, 2^{k-8}-1$\} \\
\If{$x_T \ge 0$}
{
\If{$0\le x_T\le 1$ and $-1\le y_T \le 1$}
{
SquarePlacement(2) \\
}
\Else
{
Place a pebble $p_0$ at $P$ \\
\If{$ x_T\le 1$ and $y_T > 1$} 
{
NonSquarePlacement$(1,0)$ \\
Place a pebble $p_T$ at the foot of the perpendicular drawn from $T$ on $L_{\Delta}$. \\
}
\Else
{
NonSquarePlacement$(1,1)$ \\
Place a pebble $p_T$ at the foot of the perpendicular drawn from $T$ on $L_{\Delta+1}$. \\
}
}
}
\Else
{
\If{$-1\le x_T\le 0$ and $-1\le y_T \le 1$}
{
Place a pebble $p_0$ at $P$. \\
SquarePlacement(1) \\
}
\Else
{
\If{$-1\le x_T\le 0$ and $y_T > 1$} 
{
NonSquarePlacement(2,0) \\
Place a pebble $p_T$ at the foot of the perpendicular drawn from $T$ on $L_{\Delta}$. \\
}
\Else
{
NonSquarePlacement(2,1) \\
Place a pebble $p_T$ at the foot of the perpendicular drawn from $T$ on $L_{\Delta+1}$. \\
}
}
}

\end{algorithm2e}


\begin{algorithm2e}

\caption{\sc{SquarePlacement($count$)}}\label{alg:cap01}
Place a pebble $p_1$ at $((-1)^{count},0)$. \\
\If{$q_2<1$}
{
Place the pebbles $p_2$ at $h$ and $p_3$ at $s'$, respectively. \\
}
\Else
{
Place a pebble $p_2$ at $((-1)^{count}\cdot2,0)$ and $p_3$ at $s$, respectively. \\
}

\end{algorithm2e}


\begin{algorithm2e}

\caption{\sc{NonSquarePlacement($count,bit$)}}\label{alg:cap02}
Initially $l=2$. \\
Place the pebbles $p_1$ at $((-1)^{count},0)$, $p_2$ at $((-1)^{count},(-1)^{count})$ and $p_3$ at $((-1)^{count}\cdot 2,(-1)^{count})$, respectively. \\
$\mu_j$ be the binary representation of the integer $j$ with leading $\lfloor\log k \rfloor-\lfloor\log j\rfloor$ many zeroes. \\
$\mu_j=bit.\mu_j$ \Comment{Represents the concatenation of $bit$ value with $\mu_j$}\\
\If{$bit=1$}
{
Place a pebble at $((-1)^{count}\cdot 3,0)$. \\
}
\Else
{
Place a pebble at $((-1)^{count}\cdot 4,0)$. \\
}
\While{$l\leq k+1$}
{
\If{$\ell$-th bit of $\mu_j$ is $1$}
{
Place a pebble at $((-1)^{count}\cdot (2\ell+1),0)$. \\
}
\Else
{
Place a pebble at $((-1)^{count}\cdot(2\ell+2),0)$. \\
}
$l=l+1$ \\
}
\If{1st bit of $\mu_j$ is 1}
{
Place the pebbles $p_{t_1}$ at $((-1)^{count}\cdot 4,(-1)^{count})$ and $p_{t_2}$ at $((-1)^{count}\cdot(2|\mu_j|+6),0)$, respectively. \\
}
\Else
{
Place the pebbles $p_{t_1}$ at $((-1)^{count}\cdot 5,(-1)^{count})$ and $p_{t_2}$ at $((-1)^{count}\cdot (2|\mu_j|+7),0)$, respectively. \\
}
\end{algorithm2e}


\begin{algorithm2e}

\caption{\sc{AgentMovement}}\label{alg:cap1}

If a pebble is found at $P$ then set
$angle=\pi$ otherwise set $angle=0$.\\ 
$t=2$, $\mu=\epsilon$ \\
Start moving at an angle $angle$ with the positive $x$ axis. \label{step1} \\
\If{treasure is found}
{
Terminate
}
\Else
{
Continue moving in the same direction until the $t$-th pebble or the treasure is found. \\
\If{treasure found }
{
Terminate \\
}
\Else
{
$\ell$=FindBit($t,angle$) \\
\If{$\ell\in \{0,1\}$}
{
$\mu=\mu\cdot \ell$. \\
$t=t+1$. \\
Go to Step \ref{step1}
}
\Else
{
FindTreasure($\mu,angle$) \\
}
}
}

\end{algorithm2e}

 
\begin{algorithm2e}

\caption{\sc{FindBit($t,angle$)}}\label{alg:cap2}

Move at an angle $\pi-\theta_{t}$, where $\theta_{t}=\arctan(\frac{-1}{2t})$ until the treasure or a pebble is found. \\
\If{treasure found}
{
Terminate \\
}
\Else
{
Move at an angle $\pi-\beta_{t}$, where $\beta_{t}=\arctan(\frac{2t+1}{2t-1})$. \\
\If{treasure found}
{
Terminate
}
\ElseIf{$P$ is found}
{
return 1 \\
}
\ElseIf{a pebble is found at a point other than $P$}
{
\If{$angle=0$}
{
Move at an angle $\pi+\frac{\pi}{4}$. \\
}
\Else
{
Move at an angle $\pi-\frac{\pi}{4}$. \\
}
\If{$P$ is found}
{
return 0 \\
}
\Else
{
Continue its movement until $P$ is reached. \\
return 2 \\
}
}
}

\end{algorithm2e}


\begin{algorithm2e}

\caption{\sc{FindTreasure($\mu,angle$)}}\label{alg:cap4}

Let $\Delta$ be the integer whose binary representation is $\mu$  \\
\If{$angle=\pi$}
{
\If{$\mu_1=0$}
{
$val=\Delta$ \\
SectorTravel$(val,1,2)$ \\
}
\Else
{
$val=\Delta+1$ \\
SectorTravel$(val,1,1)$ \\
}
}
\Else
{
\If{$\mu_1=0$} 
{
$val=\Delta$ \\
SectorTravel$(val,2,1)$ \\
}
\Else
{
$val=\Delta+1$ \\
SectorTravel$(val,2,2)$ \\
}
}

\end{algorithm2e}


\begin{algorithm2e}

\caption{\sc{SectorTravel($val,count,num$)}}\label{alg:cap3}

Move at an angle $\frac{\pi}{2}+(-1)^{count}\left(\frac{\pi\cdot val}{2^{k-8}}\right)$ until a pebble or treasure is found. \\
\If{Treasure found}
{
Terminate. \\
}
\Else
{
Move at an angle $\pi+(-1)^{num}\frac{\pi}{2}$ until treasure is found. \\
Terminate. \\
}

\end{algorithm2e}

\subsection{Treasure hunt}\label{treasurehunt}
 
 Starting from $P$, the agent finds the treasure with the help of the pebble placed at different points on the plane. On a high level, the agent performs three major tasks: (1) Learn the direction of its initial movement (2) Learn the encoding of the sector number in which the treasure is located, and (3) Move inside the designated sector and find the treasure.
 
 The agent learns the direction of its initial movement by observing whether a pebble is placed at $P$ or not. If a pebble is placed, then it learns that the direction of its initial movement is west and pebble placement is done for the encoding of the sector number on the negative $x$ axis. Otherwise,  it learns that the direction of its initial movement is east and pebble placement is done for the encoding of the sector number on the positive $x$ axis. Then for each $j=1,2,\cdots$, it continues its movement in a specific path (depending on the value of $j$) and learns the $j$-th bit of the encoding until it detects the termination of the encoding. To be specific, the $j$-th bit of the encoding is learned by the agent using the movements in the following order from $P$. 
 
 \begin{itemize}
 \item Starting from $P$, move along $x$-axis until the $(j
 +1)$-th pebble is reached, 
 \item Move at angle $\arctan({\frac{-1}{2j}})$, and continue moving in this direction until a pebble is reached
 \item Move at an angle $\arctan({\frac{2j+1}{2j-1}})$ until $P$ or a pebble is found.
 \item If $P$ is found in the previous step, then the bit is 1.
 \item If a pebble is found, then move along $x$ axis towards $P$. If $P$ is encountered, then the bit is 0.
 \item If a pebble is encountered instead of $P$ in the previous step, then the agent learns that the encoding is completed.
 
\end{itemize} 
 Let $\mu$ be the binary string learned by the agent in the above process and let $\Delta$ be the integer whose binary representation of $\mu$. If the first bit of $\mu$ is 1, then the agent starts moving along $L_{\Delta+1}$ from $P$ until it hits a pebble or reaches the treasure. Once the pebble is reached, the agent changes its direction at angle $\frac{\pi}{2}$ if its initial direction of movement was west, else the agent changes its direction at angle $\frac{3\pi}{2}$. It continues its movement in the current direction until the treasure is reached.

The following lemma ensures that the pebbles are placed at a distance of at least 1 in step 4 of Case 1 in the above pebble placement strategy.

\begin{lemma}\label{noconflict}
If $T=(x_T,y_T)\in B'$, where $B'=\{(x,y)|$ $ 0\leq x \leq 1$ and $y_T>1\}$, the location of the foot of the perpendicular $F$ on $L_j$ is outside the square $B$.
\end{lemma}

\begin{proof}
Let the position of $F$ be $(h,k)$. Let $m_1$ be the slope of the line $PF$ and $m_2$ be the slope of the line $FT$ (See Fig. \ref{slope}). Now as $PF \perp FT$, therefore $m_1\cdot m_2=-1$. The slope $m_1= \frac{k}{h}$ is positive as $k>0$ and $h>0$, so $m_2= \frac{y_T-k}{x_T-h}$ must be negative to satisfy the above condition. Now, $m_2$ can be negative if one of the following cases is true.
\begin{itemize}
\item Case-1: $y_T>k$ and $x_T<h$, 
    \item Case-2: if $y_T<k$ and $x_T>h$,
    
\end{itemize}
If Case 1 is true, then the point $F$ must be on the right side of the line $x=X_T$, which is not possible. Therefore, Case 2 must be true, i.e., $1<y_T<k$ and $x_T>h$. This implies that $F$ is outside $B$.
\end{proof}
\qed

 The execution of the algorithm is explained with the help of a example.

 \begin{figure}[h]
  \centering
  \includegraphics[width=0.7\linewidth]{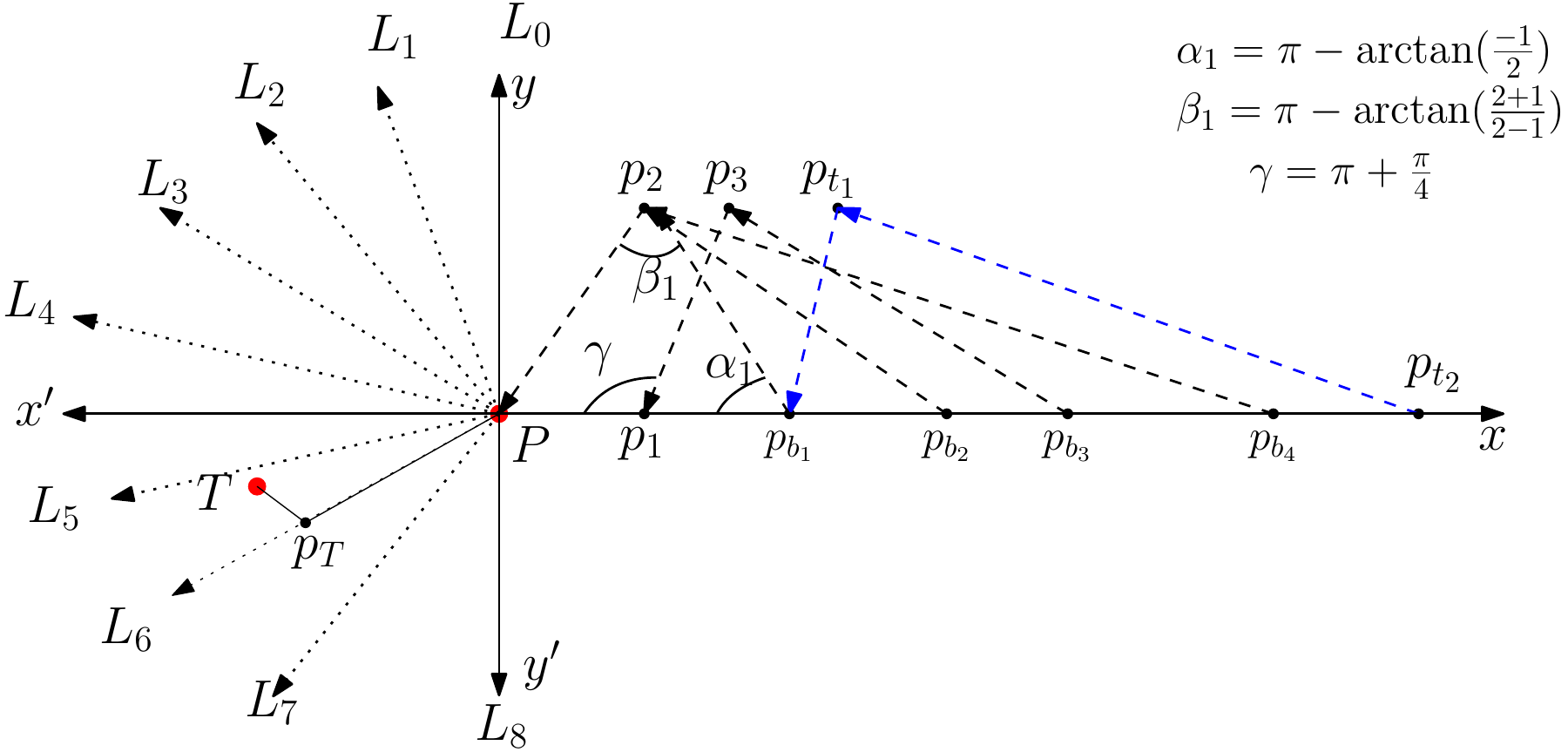}
  \caption{Figure showing demonstration of Example-1}
  \label{ex-travel}
\end{figure}

 \noindent\textbf{Example-1:} Given $11$ pebbles, the oracle divides the plane into $2^{11-8}$ sectors. Suppose the treasure is placed in the sector $5$ (as depicted in Fig. \ref{ex-travel}). Moreover, consider the position of the treasure is outside the square $B$. So, the oracle places the pebbles by following the algorithm \textit{PebblePlacement \ref{alg:cap0}}, such that the agent, after following the algorithm \textit{AgentMovement \ref{alg:cap1}} learns the direction of its initial movement and further learns the encoding of the sector number (i.e., 101 in this case) in which the treasure is located in the following manner. 
 
 An iteration of the algorithm \ref{alg:cap1} is defined as a cycle which consists of the agent's movement starting from $P$ and returning back to $P$.
 In the first iteration, the agent initially at $P$, does not find a pebble at $P$. The algorithm \ref{alg:cap1} instructs the agent to move towards east until it encounters a second pebble $p_{b_1}$ along positive $x$-axis. From $p_{b_1}$ the agent moves at an angle $\pi-\arctan(\frac{-1}{2})$ until it encounters a pebble $p_2$. From $p_2$ it further moves at an angle $\pi-\arctan \left(\frac{2+1}{2-1}\right)$ until it reaches the origin $P$. So, after completion of the first iteration (i.e., the path traversed $P\rightarrow p_{b_1}\rightarrow p_2\rightarrow P$) the agent learns that the first bit is $1$.
 In the second iteration, the agent again moves towards east until it reaches the third pebble $p_{b_2}$. From $p_{b_2}$, the agent moves at an angle $\pi-\arctan\left(\frac{-1}{2\cdot 2}\right)$ until it encounters a pebble $p_2$, from $p_2$ it further moves at an angle $\pi-\arctan\left(\frac{2\cdot 2+1}{2\cdot 2-1}\right)$ until it reaches the origin $P$. So, after completion of the second iteration (i.e., the path traversed $P\rightarrow p_{b_2}\rightarrow p_2\rightarrow P$) the agent learns that the second bit is again $1$. In the third iteration, the agent after a similar movement towards east reaches the fourth pebble $p_{b_3}$ along the positive x-axis. From $p_{b_3}$, it further moves at an angle $\pi-\arctan(\frac{-1}{2\cdot 3})$ until it reaches the pebble $p_3$. From $p_3$, it moves along an angle $\pi-\arctan \left(\frac{2\cdot 3+1}{2\cdot 3-1}\right)$ until it reaches a pebble $p_1$. From $p_1$, the agent finally moves at an angle $\pi+\frac{\pi}{4}$ until it reaches $P$. So, after completion of the third iteration (i.e., the path traversed $P\rightarrow p_{b_3}\rightarrow p_3 \rightarrow p_1\rightarrow P$) the agent learns that the third bit is 0. In the fourth iteration, with a similar movement the agent reaches the pebble $p_{b_4}$, and from this position the agent moves at an angle $\pi-\arctan\left(\frac{-1}{2\cdot 4}\right)$ until it reaches $p_2$, from $p_2$ it further moves at an angle $\pi-\arctan \left(\frac{2\cdot 4+1}{2\cdot 4-1}\right)$ until it reaches $P$. So, in the fourth iteration (i.e., the path traversed $P\rightarrow p_{b_4}\rightarrow p_2\rightarrow P$), the agent learns that the fourth bit is $1$. In the fifth iteration, the agent reaches the fifth pebble, i.e., $p_{t_2}$ (refer Fig. \ref{ex-travel}), from $p_{t_2}$ it moves at an angle $\pi-\arctan\left(\frac{-1}{2\cdot 5}\right)$ until it reaches a pebble $p_{t_1}$, from this position it further moves at an angle $\pi-\arctan \left(\frac{2\cdot 5+1}{2\cdot 5-1}\right)$ until it reaches a pebble $p_{b_1}$. Further from $p_{b_1}$, the agent further moves at an angle $\pi+\frac{\pi}{4}$ until it reaches $P$. Since the agent encounters the pebble $p_1$ after its last movement from $p_{b_1}$, this gives the knowledge to the agent that termination is achieved. Hence the binary string obtained by the agent is $\mu=1101$(say).
 
 Then, the agent by following the algorithm \textit{AgentMovement \ref{alg:cap1}} decodes that somewhere in sector $5$ the treasure is located. Further, since $\mu_1=1$, the agent then follows the algorithms \textit{FindTreasure \ref{alg:cap4}} and \textit{SectorTravel \ref{alg:cap3}} to finally reach the treasure by traversing the half-line $L_6$ and encountering the pebble $p_T$, from which moving at an angle $\pi+\frac{\pi}{2}$ the agent ultimately reaches the treasure $T$.

\section{Complexity}\label{correct}
In this section, we give the correctness and an upper bound on the cost of finding treasure from the proposed algorithms.

The following two lemmas show the algorithm's correctness when the treasure is inside $B$ and the upper bound of the cost of treasure hunt.

\begin{lemma}\label{Box-correct}
With 3 pebbles and the treasure located inside $B$, the agent successfully finds the treasure.
\end{lemma}

\begin{proof} 
When the treasure is present inside the square $B$, the oracle places a pebble $p_0$ at $P$, if the treasure is located in the left half of $y$-axis. Otherwise, no pebble is placed at $P$ as discussed in Case-2 of section \ref{pebbleplacement} (also refer to lines 3 and 14 of algorithm \ref{alg:cap0}). So, the agent starts its movement from $P$ along an angle $\pi$, i.e., along negative $x$-axis if it finds a pebble $p_0$ at $P$ (refer lines 1 and 2 of algorithm \ref{alg:cap1}) otherwise, the agent moves along an angle $0$, i.e., positive $x$-axis if no pebble is found at $P$ (refer line 4 of algorithm \ref{alg:cap1}). Now we have the following cases depending on the presence of a pebble at $P$.
\begin{itemize}
    \item \textit{Pebble not found at $P$}: In this case, the agent while moving along positive $x$-axis either finds the treasure and the algorithm terminates (refer to lines 7 and 8 of Algorithm \ref{alg:cap1}). Otherwise, the agent finds a pebble $p_1$ placed at $(1,0)$, which it ignores as instructed in the algorithm \textit{FindBit} \ref{alg:cap2} and continues to move until it reaches the treasure or encounters a pebble. If the treasure is not found, then the pebble $p_2$ is placed by the oracle at either $h$ or $(2,0)$ (refer to lines 10 and 13 of algorithm \ref{alg:cap01}). The agent, after encountering the second pebble, moves along an angle $\pi-\theta_1$, where $\theta_1=\arctan \left(\frac{-1}{2}\right)$ until the treasure or a pebble is found. If a pebble is found, then we have the following cases:
    \begin{itemize}
        \item \textit{$p_2$ placed at (2,0)}: In this case, the agent finds the pebble $p_3$ at $s$ (refer to line 14 of algorithm \ref{alg:cap01}), from which it further moves along an angle $\pi-\beta_{1}$, where $\beta_1=\arctan(3)$ and finds the treasure.
        \item \textit{$p_2$ placed at $h$}: In this case, the agent finds the pebble $p_3$ at $s'$ (refer to line 11 of algorithm \ref{alg:cap01}), from which it further moves along an angle $\pi-\beta_{1}$, where $\beta_1=\arctan(3)$ and finds the treasure.
    \end{itemize}
    \item \textit{Pebble found at $P$}: In this case, the agent moves along negative $x$-axis and performs the similar task as described above. The reason being, the pebbles are placed in a similar manner, just on the adjacent half (i.e., left half of $y$-axis) as discussed in the above case (refer to $x_T<0$ in case 2 of section \ref{pebbleplacement}). 
\end{itemize}
\qed
\end{proof}

\begin{figure}[h]
  \centering
  \includegraphics[width=0.55\linewidth]{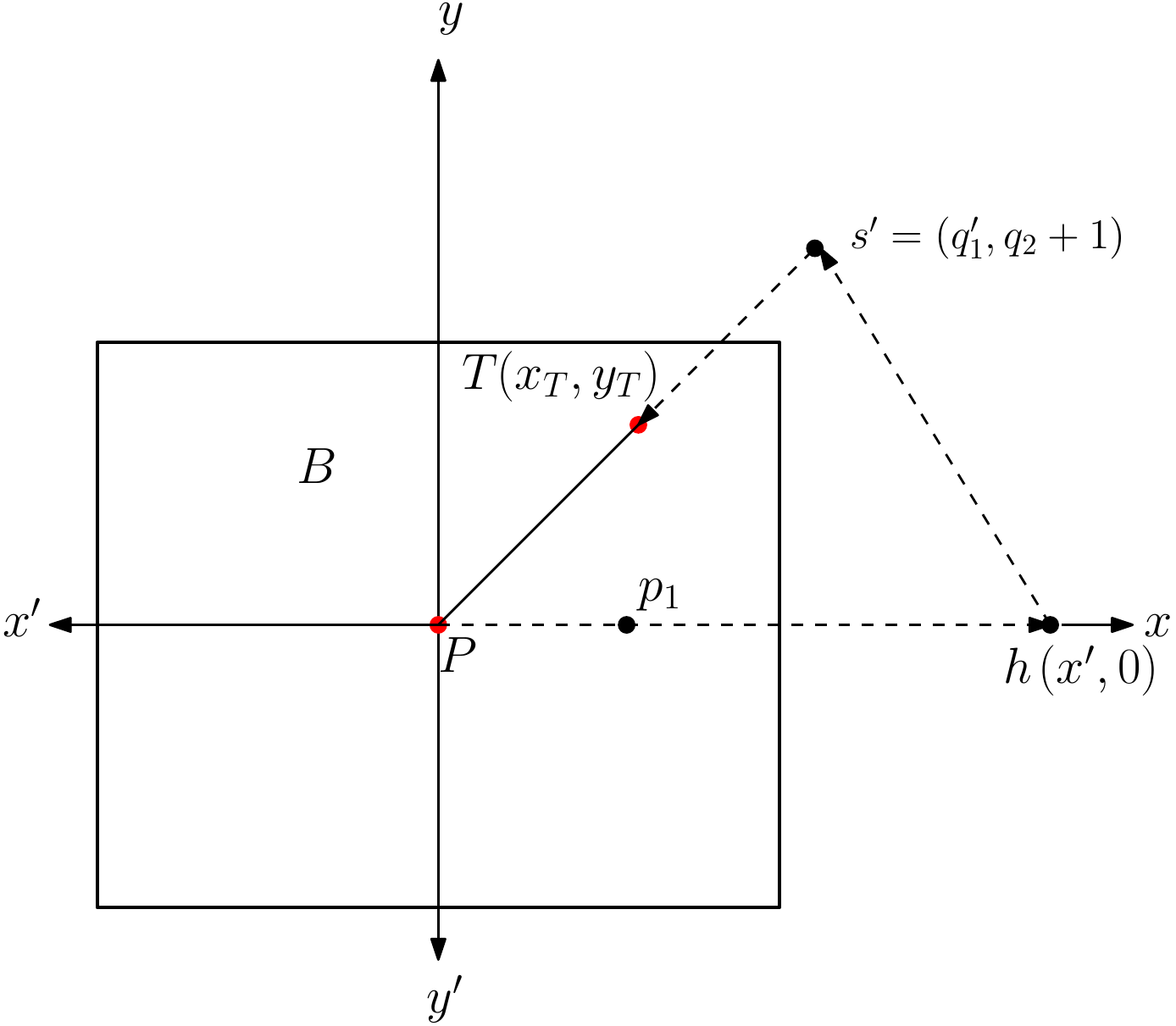}
  \caption{Traversal of an agent when the treasure is inside the square}
  \label{square-c}
\end{figure}

\begin{lemma}\label{Box-complex}
When the treasure is located inside $B$, the agent starting from $P$ successfully finds the treasure at cost $O(D)$.
\end{lemma}

\begin{proof}
The treasure is located at $(x_T,y_T)$ and let the co-ordinates of $h$ be $(x',0)$.
The worst possible placement of pebbles $p_2$ and $p_3$ by the oracle are at $h$ and $s'$ (where $s'=(q'_1,q_2+1)$ refer Fig. \ref{square-c} and refer to the lines 10 and 11 of algorithm \ref{alg:cap01}), respectively. So, the traversal of the agent to reach the treasure will be along the path $Ph\rightarrow hs'\rightarrow s'T$. The total cost of this traversal is as follows: $x'+\sqrt{(x'-q'_1)^2+(q_2+1)^2}+\sqrt{(q'_1-x_T)^2+(q_2+1-y_T)^2}$ (where $|Ph|=x'$, $|hs'|=\sqrt{(x'-q'_1)^2+(q_2+1)^2}$ and $|s'T|=\sqrt{(q'_1-x_T)^2+(q_2+1-y_T)^2}$). Now, since $|PT|< |Ph|+|hs'|+|s'T|$, i.e., $D< |Ph|+|hs'|+|s'T|$ (since $|PT|=D$). Hence, in this case, the cost of reaching the treasure is $O(D)$.  
\qed
\end{proof}

Lemma \ref{NotBoxbit-cost}, Lemma \ref{NotBox-binary}, Lemma \ref{NotBox-Sectorcorrect} and Lemma \ref{NotBox-sector} shows the correctness of the algorithm when the treasure is located outside $B$.

\begin{lemma}\label{NotBoxbit-cost}
When the treasure is outside $B$, the agent successfully finds the $j$-th bit of the binary string $\mu$ at cost $O(j)$.
\end{lemma}

\begin{proof}
To obtain the $j$-th bit of $\mu$ the movement of the agent is as follows.

When the treasure is present outside the square $B$, the oracle places a pebble $p_0$ at $P$ if the treasure is located in the right half of $y$-axis otherwise, there is no pebble placed at $P$ as discussed in case 1 of section \ref{pebbleplacement} (refer to line 6 of algorithm \ref{alg:cap0}). The movement of the agent from $P$ is as follows:
\begin{itemize}
    \item \textbf{$p_0$ found at $P$}: In this case the agent moves at an angle $\pi$, i.e., along negative $x$-axis (refer to the lines 1, 2 and 6 of algorithm \ref{alg:cap1}). Further, it ignores the first $j$ many pebbles along the negative $x$-axis (refer to line 10 of algorithm \ref{alg:cap1}) and moves until it either finds the treasure or encounters the $(j+1)$-th pebble $p_{b_j}$ or $p_{t_1}$ placed at either $(-2j-1,0)$ or $(-2j-2,0)$ or $(-2j-6,0)$ or $(-2j-7,0)$. If the treasure is not found, the cost of reaching this pebble is at most $2j+7$. Now, from the current position, the agent is instructed to move at an angle $\pi-\theta_j$ (refer to line 1 of algorithm \textit{FindBit} \ref{alg:cap2}), where $\theta_j=\arctan \left(\frac{-1}{2j}\right)$ until the treasure or a pebble $p_2$ or $p_3$ or $p_{t_2}$ is encountered. If a pebble is found, then this pebble is either $p_2$ placed at (-1,-1) or $p_3$ placed at (-2,-1) or $p_{t_2}$ placed at either (-4,-1) or (-5,-1), respectively. So, the cost of this traversal from the $(j+1)$-th pebble to either $p_2$ or $p_3$ or $p_{t_2}$ is at most $\sqrt{{(2j+2)}^2+1}$. From either of these pebbles, the agent is further instructed to move along an angle $\pi-\beta_{j}$, where $\beta_j=\arctan \left(\frac{2j+1}{2j-1}\right)$ (refer to line 5 of algorithm \ref{alg:cap2}) until it encounters the treasure or encounters a pebble or reaches $P$ with $O(1)$ cost. Now we have the following cases:
    \begin{itemize}
        \item \textit{If treasure found}: In this case, the agent has reached its goal, and the whole process terminates.
        \item \textit{If pebble found}: In this case, the pebble found is either $p_1$ or $p_{b_1}$. In either of the case, the agent is further instructed to move along an angle $\pi+\frac{\pi}{4}$ or $\pi-\frac{\pi}{4}$ (refer to the lines 12 and 14 of algorithm \ref{alg:cap2}) until it reaches $P$ or a pebble is found. Hence we have two cases:
        \begin{itemize}
            \item \textit{If $P$ reached}: The agent gains the information that the $j$-th bit of $\mu$ is 0 (refer to the lines 15 and 16 of algorithm \ref{alg:cap2} and lines 14, 15 and 16 of algorithm \ref{alg:cap1}). So, the path traveled to gain this information is $P\rightarrow p_{b_j} \rightarrow p_3 \rightarrow p_1 \rightarrow P$. So, the cost of this traversal is at most $(2j+2) + \sqrt{{(2j)}^2+1} + O(1)$, which is $O(j)$.
            \item \textit{If pebble found}: In this case, the agent continues to move until $P$ is reached and in which case the agent gains the information that termination is achieved, i.e., $(j-1)$-th bit is the terminating bit of $\mu$. The agent further moves on to execute algorithm \textit{FindTreasure} (refer to line 18 of algorithm \ref{alg:cap2}, and to lines 15 and 20 of algorithm \ref{alg:cap1}). So, the path travelled to gain this information is $P\rightarrow p_{t_1} \rightarrow p_{t_2} \rightarrow p_{b_1} \rightarrow p_1\rightarrow P$. So, the cost of this traversal is at most $(2j+7) + \sqrt{{(2j+2)}^2+1} + O(1)$,  which is $O(j)$.
        \end{itemize}
        \item \textit{If $P$ is reached}: In this case the agent gains the information that the $j$-th bit of $\mu$ is 1 (refer to the lines 8 and 9 of algorithm \ref{alg:cap2} and lines 14, 15 and 16 of algorithm \ref{alg:cap1}). So, the path traveled to gain this information is $P\rightarrow p_{b_j} \rightarrow p_2 \rightarrow P$. So, the cost of this traversal is at most $(2j+1) + \sqrt{{(2j)}^2+1} + O(1)$,  which is again $O(j)$.
    \end{itemize}
  
    \item \textbf{No pebble found at $P$}: In this case, the agent moves in a similar manner as the pebbles are placed in a similar way as for each pebble placed at $(m,n)$, where $m\neq 0$ and $n\neq 0$ for the above case, the oracle places the corresponding pebble at $(-m,-n)$. Similarly, the cost to obtain the $j$-th bit of the binary string is $O(j)$.
\end{itemize}

  Hence in each case, the cost of finding the $j$-th bit of $\mu$ is $O(j)$.

\qed
\end{proof}

\begin{lemma}\label{NotBox-binary}
Given $k$ pebbles and the treasure located outside $B$, the agent successfully finds the binary string $\mu$ at cost $O(k^2)$. 
\end{lemma}

\begin{proof}
According to lemma \ref{NotBoxbit-cost}, the agent successfully determines the $j$-th bit value of $\mu$ at $O(j)$ cost. Now as the binary string, $\mu$ is of length $k$ this implies, the total cost to obtain $\Gamma$ is $ \sum_{j=1}^{k} O(j) $, i.e., $O(k^2)$.
\qed
\end{proof}
\begin{figure}
\centering
\begin{minipage}{.45\textwidth}
  \centering
  \includegraphics[width=1.1\linewidth]{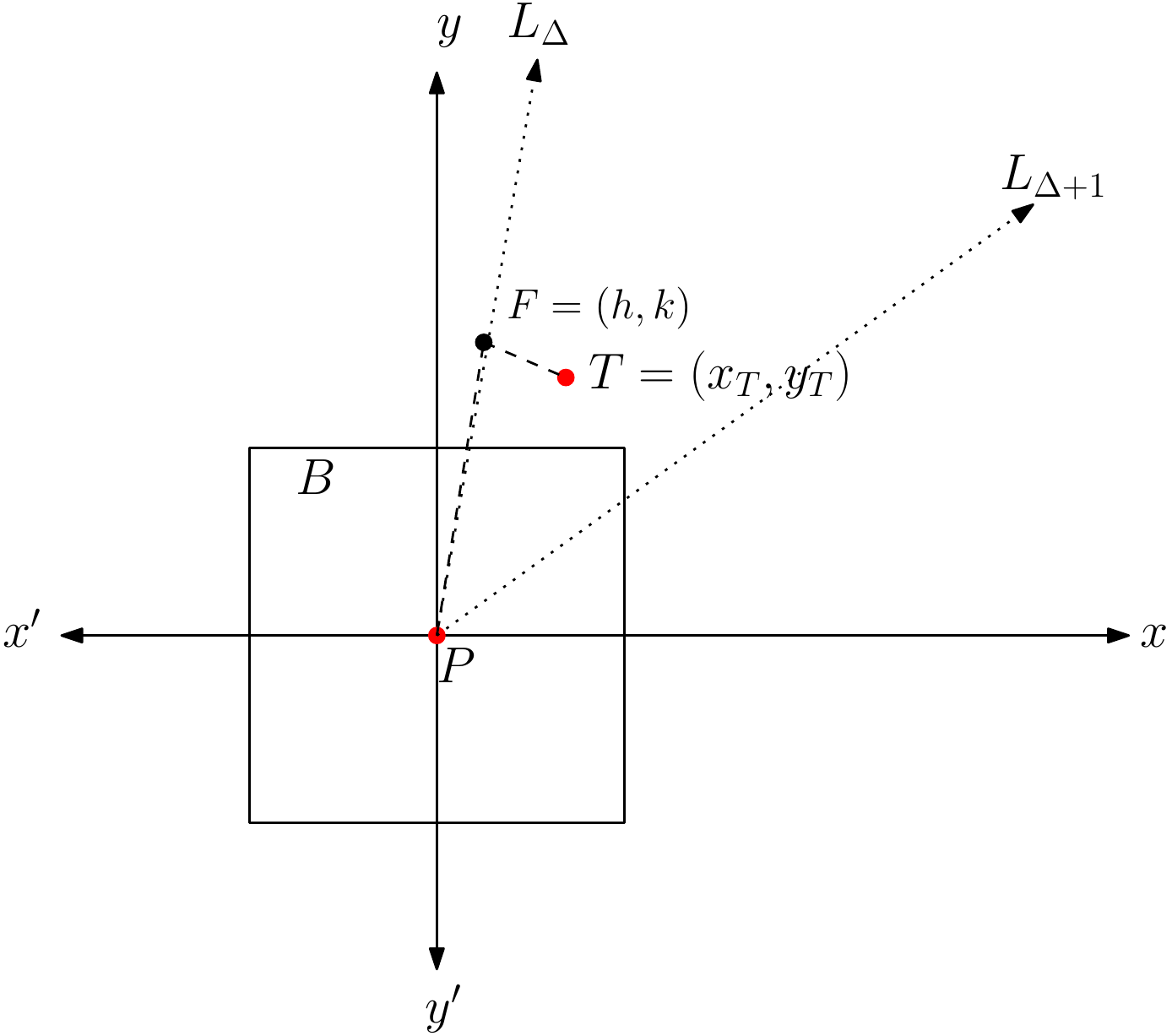}
  \captionof{figure}{To determine the location of pebble $p_T$}
  \label{slope}
\end{minipage}%
\hfill
\begin{minipage}{.5\textwidth}
  \centering
  \includegraphics[width=.8\linewidth]{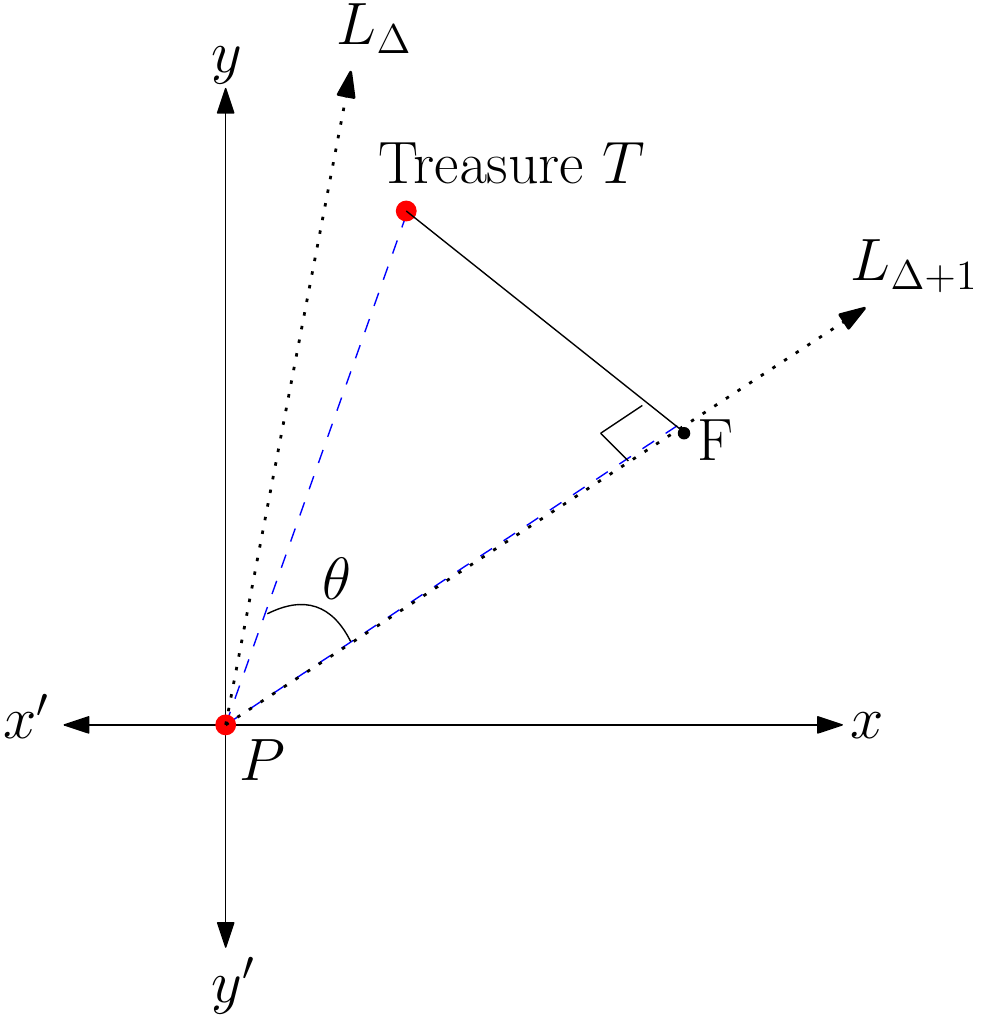}
  \captionof{figure}{Traversal of the agent inside the sector when $\mu_1=1$ of $\Gamma$}
  \label{final}
\end{minipage}
\end{figure}

\begin{lemma}\label{NotBox-Sectorcorrect}
When the treasure is located outside $B$, the agent after gaining the binary string $\mu$, successfully finds the treasure by executing the algorithm \textit{FindTreasure}.
\end{lemma}

\begin{proof}
After termination of algorithm \textit{AgentMovement} \ref{alg:cap1}, the agent performs the algorithm \textit{FindTreasure} \ref{alg:cap4} with the already acquired binary string $\mu$ to finally reach the treasure $T$ if not already reached. 
 
 The treasure is either located somewhere on the region $x\ge 0$ (i.e., right half of $y$-axis) or $x\le 0$ (i.e., left half of $y$-axis) and accordingly, the oracle divides the whole left half or right half of $y$-axis into $2^{k-8}$ sectors (refer to line 1 of algorithm \ref{alg:cap0}), where a sector $S_i$ is bounded by half-lines $L_i$ and $L_{i+1}$ and angle between consecutive half lines is $\frac{\pi}{2^{k-8}}$. Suppose the treasure is located somewhere in sector $S_\Delta$, so $\mu$ is the binary representation of $\Delta$. The agent decodes this value $\Delta$ after executing the algorithm {\it AgentMovement} \ref{alg:cap1}. The whole aim of the oracle is to align the agent either along the half-line $L_\Delta$ or $L_{\Delta+1}$. The alignment of the agent along the half-lines $L_\Delta$ or $L_{\Delta+1}$ depends on the first bit value of $\mu$, i.e., on $\mu_1$ (refer to line 3 in algorithm \ref{alg:cap4}) in the following manner:
 \begin{itemize}
     \item \textit{Case $\mu_1=0$}: If a pebble is found at $P$ (i.e., $angle=\pi$ refer to line 2 of algorithm \ref{alg:cap1}) then the agent is instructed to move along an angle $\frac{\pi}{2}-\frac{\pi\Delta}{2^{k-8}}$, i.e., along the half-line $L_{\Delta}$ until the treasure or a pebble is found (refer to the lines 4 and 5 of algorithm \ref{alg:cap4} and line 1 of algorithm \ref{alg:cap3}). Otherwise, if no pebble is found at $P$ (i.e., $angle = 0$ refer to line 4 of algorithm \ref{alg:cap1}) then the agent is instructed to move along an angle $\frac{\pi}{2}+\frac{\pi\Delta}{2^{k-8}}$ (refer to lines 11 and 12 of algorithm \ref{alg:cap4} and to line 1 of algorithm \ref{alg:cap3}) until it finds the treasure or a pebble.
     \begin{itemize}
         \item \textit{If treasure found}: Then the algorithm terminates as we have reached our goal (refer to the lines 2 and 3 of algorithm \ref{alg:cap3}).
         \item \textit{If pebble found}: The agent is further instructed to move along an angle $\pi+\frac{\pi}{2}$ or $\pi-\frac{\pi}{2}$ depending on the angle $\pi$ or $0$ (refer to line 5 of algorithm \ref{alg:cap3}) until treasure is found.
     \end{itemize}
     \item \textit{Case $\mu_1=1$}: If a pebble is found at $P$ (i.e., $angle=\pi$) then the agent is instructed to move along an angle $\frac{\pi}{2}-\frac{\pi(\Delta+1)}{2^{k-8}}$, i.e., along the half-line $L_{\Delta+1}$ until the treasure or a pebble is found (refer to the lines 7 and 8 of algorithm \ref{alg:cap4} and line 1 of Algorithm \ref{alg:cap3}). Otherwise, if no pebble is found at $P$ (i.e., $angle =0$), then the agent is instructed to move along an angle $\frac{\pi}{2}+\frac{\pi(\Delta+1)}{2^{k-8}}$ (refer to the lines 14 and 15 of algorithm \ref{alg:cap4} and line 1 of algorithm \ref{alg:cap3}) until it finds the treasure or a pebble.
     \begin{itemize}
         \item {\it If treasure found}: Then the algorithm terminates as we have reached our goal (refer to the lines 2 and 3 of algorithm \ref{alg:cap3}).
         \item {\it If pebble found}: The agent is further instructed to move along an angle $\pi-\frac{\pi}{2}$ or $\pi+\frac{\pi}{2}$ depending on the angle $\pi$ or $0$ (refer to line 5 of algorithm \ref{alg:cap3}) until the treasure is found.
     \end{itemize}
 \end{itemize}
 Hence in each case, the agent successfully finds the treasure after executing the algorithm \textit{FindTreasure}.

\qed 
\end{proof}
\begin{figure}[h]
\centering
  \includegraphics[width=0.7\linewidth]{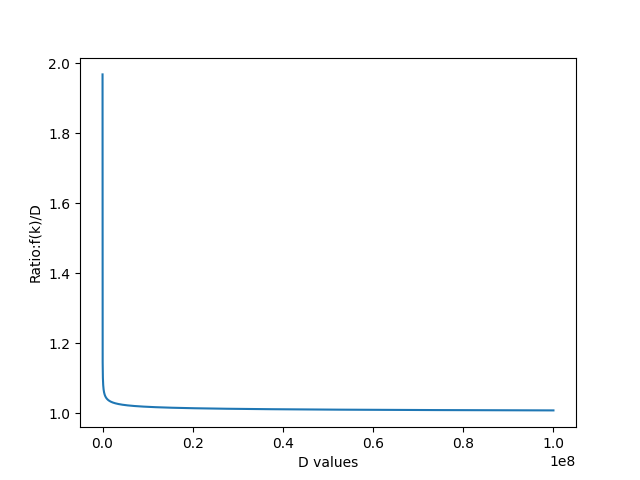}
  \caption{The curve represents the ratio of $\frac{f(k)}{D}$ for different values of $D$}
  \label{ratio}
\end{figure}   

\begin{lemma}\label{NotBox-sector}
When the treasure is located outside $B$, the agent after gaining the binary string $\mu$ finds the treasure at cost $D(\sin\theta' + \cos\theta')$, where $\theta' = \frac{\pi}{2^{k'}}$ and $k'=k-8$.
\end{lemma}

\begin{proof}
The agent, after gaining the binary string $\mu$, executes the algorithms {\it AgentMovement} and {\it FindTreasure}, and successfully reaches the treasure by following the path $PF\rightarrow FT$ from $P$ (refer Fig. \ref{final}). Since the angle between $L_\Delta$ and $L_{\Delta+1}$ is $\frac{\pi}{2^{k-8}}$. Hence, $\angle FPT$ is at most $\frac{\pi}{2^{k-8}}$ (refer Fig. \ref{final}( which is $\theta'$ (say), also $\angle TFP=\frac{\pi}{2}$ (as $F$ is the foot of perpendicular of $T$ to $L_{\Delta+1}$ if $\mu_1=1$ otherwise if $\mu_1=0$ then $F$ is the foot of perpendicular of $T$ to $L_{\Delta}$) and $|PT|\leq D$. So we have $PF=D\cos\theta'$ and $FT=D\sin\theta'$. Hence, the cost of traveling along the sector $S_{\Delta}$ from $P$ to reach $T$ is $PF+FT$, i.e.,  $D(\sin\theta' + \cos\theta')$.
 \qed 
\end{proof}

Combining Lemma \ref{NotBox-Sectorcorrect} to Lemma \ref{NotBox-sector}, we have the final result of this section summarized by the following theorem. 

\begin{theorem}
Given $k$ pebbles, the agent starting from $P$ successfully finds treasure with $O(k^2)+D(\sin\theta' + \cos\theta')$-cost, where $\theta' = \frac{\pi}{2^{k'}}$ and $k'=k-8$.
\end{theorem}

\begin{remark}
Consider the function $f(k)=O(k^2)+D(\sin\theta'+\cos\theta')$, where $\theta'=\frac{\pi}{2^{k-8}}$. Note that  for $D,k \rightarrow \infty$ and $k \in o(\sqrt{D})$, the value of $\frac{f(k)}{D} \rightarrow 1$. In order to demonstrate this fact, we plot the value of $\frac{f(k)}{D}$ for increasing values of $D$ in the range $[1000,100000000]$ and for $k=\lfloor D^{\frac{1}{3}}\rfloor$. Fig. \ref{ratio} shows the values of $\frac{f(k)}{D}$ for different values of $D$ in the above mentioned range and for the fix value of $k$ for each $D$. This figure shows that for large value of $D$, the value of $\frac{f(k)}{D}$ is very close to 1.
\end{remark}

\section{Conclusion}\label{section-4}

We propose an algorithm for the treasure hunt that finds the treasure in an Euclidean plane using $k\ge 9$ pebbles at cost $O(k^2)+D(\sin\theta'+\cos\theta')$, where $\theta'=\frac{\pi}{2^{k-8}}$. Proving a matching lower bound remains an open problem to consider in the future.
It can be noted that if the agent has some visibility, the problem becomes very trivial even with only one pebble: place a pebble on the line from $P$ to $T$ within a distance of $r$ from $P$, where $r$ is the visibility radius of the agent. Starting from $P$, the agent sees the position of the pebble, move to the pebble, and then continue until it hits the treasure. But the problem becomes challenging if the compass of the agent is not perfect: i.e., if the agent does not have ability to measure an angle accurately. This seems a nice future problem as an extension of the current work.  
\bibliographystyle{splncs03}

\bibliography{bibliog}
\end{document}